\documentclass[cis,nocrop]{ipart_v1_arxiv}

\firstpage{1}
\Vol{XX}
\Issue{X}
\Year{2020}

\usepackage{tikz}
\usetikzlibrary{matrix,decorations.pathreplacing}
\usepackage{amsmath}
\usepackage{amsfonts,amssymb}
\usepackage{citesort}
\usepackage{mathtools}
\usepackage{tcolorbox}

\newcommand{\N}{\mathbb{N}}

\newcommand{\R}{\mathbb{R}}
\newcommand{\Q}{\mathbb{Q}}

\theoremstyle{plain}
\newtheorem{theorem}{Theorem}

\newtheorem{corollary}[theorem]{Corollary}

\theoremstyle{definition}
\newtheorem{definition}[theorem]{Definition}

\theoremstyle{definition}
\newtheorem{remark}[theorem]{Remark}

\renewcommand{\vec}[1]{\boldsymbol{#1}}

\newcommand{\sA}{\mathcal{A}}
\newcommand{\sM}{\mathcal{M}}
\newcommand{\sP}{\mathcal{P}}
\newcommand{\sQ}{\mathcal{Q}}
\newcommand{\sS}{\mathcal{S}}
\newcommand{\sU}{\mathcal{U}}
\newcommand{\sX}{\mathcal{X}}
\newcommand{\sY}{\mathcal{Y}}

\newcommand{\fT}{\mathfrak{T}}

\newcommand{\Cl}{\underline{C}}
\newcommand{\Cu}{\overline{C}}
\newcommand{\sPc}{\mathcal{P}_c}
\newcommand{\sQc}{\mathcal{Q}_c}
\newcommand{\pone}{p^{(1)}}
\newcommand{\ptwo}{p^{(2)}}
\newcommand{\qone}{q^{(1)}}
\newcommand{\qtwo}{q^{(2)}}

\begin{document}

\title[Non-Computability and Non-Approximability of  the FSC Capacity]{Shannon Meets Turing: Non-Computability and Non-Approximability of  the Finite State Channel Capacity\\
\small \emph{In honor of Prof. Thomas Kailath on the occasion of his 85th birthday}}

\author[H.~Boche, R.~F.~Schaefer and H.~V.~Poor]{Holger Boche, Rafael F. Schaefer, and H. Vincent Poor}

\begin{abstract}
	The capacity of finite state channels (FSCs) has been established as the limit of a sequence of multi-letter expressions only and, despite tremendous effort, a corresponding finite-letter characterization remains unknown to date. This paper analyzes the capacity of FSCs from a fundamental, algorithmic point of view by studying whether or not the corresponding achievability and converse bounds on the capacity can be computed algorithmically. For this purpose, the concept of Turing machines is used which provide the fundamental performance limits of digital computers. To this end, computable continuous functions are studied and properties of computable sequences of such functions are identified. It is shown that the capacity of FSCs is not Banach-Mazur computable which is the weakest form of computability. This implies that there is no algorithm (or Turing machine) that can compute the capacity of a given FSC. As a consequence, it is then shown that either the achievability or converse must yield a bound that is not Banach-Mazur computable. This also means that there exist FSCs for which computable lower and upper bounds can never be tight. To this end, it is further shown that the capacity of FSCs is not approximable, which is an even stricter requirement than non-computability. This implies that it is impossible to find a finite-letter entropic characterization of the capacity of general FSCs. All results hold even for finite input and output alphabets and finite state set. Finally, connections to the theory of effective analysis are discussed. Here, results are only allowed to be proved in a constructive way, while existence results, e.g., proved based on the axiom of choice, are forbidden.
\end{abstract}

\maketitle

\section{Introduction}
\label{sec:introduction}

Finite state channels (FSCs) model discrete channels with memory where the channel output depends not only on the current channel input but also on the underlying channel state. The channel state allows the channel output to implicitly depend on previous channel inputs and outputs. FSCs are of significant interest as they allow to model certain types of channel variations appearing in wireless communications including e.g. flat fading and intersymbol interference \cite{Gallager-1968-InformationTheory}. FSCs are relatively simple channels and are usually used for approximations of more complex, time-continuous channels. The theory of time-continuous channels goes back to Kailath's seminal work \cite{Kailath-1959-TechRep-SamplingLinearTimeVariantFilter}. Subsequently, communication over such time-continuous channels has been studied, for example, in \cite{Kailath-1960-TIT-CorrelationDetection,Kailath-1961-Thesis-CommunicationRandomlyVaryingChannels,Kailath-1962-TIT-MeasurementsTimeVariantChannels,Kailath-1963-TIT-TimeVariantCommunicationChannels}. But FSCs are also used in molecular communication  \cite{NakanoEckfordHaraguchi-2013-MolecularCommunication}. In the latter context, the trapdoor channel has been introduced as a simple two-state channel and is studied in \cite{Blackwell-1961-InformationTheory,AhlswedeKaspi-1987-TIT-PermutingChannels,KobayashiMorita-2002-ISIT-TrapdoorChannel,Permuter-2008-TIT-TrapdoorChannelFeedback}. This channel is also known as ``chemical channel'' due to Cover. The \emph{indecomposable finite state channel} (IFSC) is introduced in \cite{BlackwellBreimanThomasian-1958-IndecomposableFSC}.  Estimating the capacity of flat fading IFSCs is considered in \cite{GoldsmithVaraiya-1996-TIT-FiniteStateMarkovChannels}. The compound capacity of FSCs is studied in~\cite{LapidothTelatar-1998-TIT-CompoundFiniteStateChannels}.  

Determining the capacity of FSCs is extremely challenging. For example, already for the trapdoor channel, the capacity is unknown. Only a lower bound \cite{KobayashiMorita-2002-ISIT-TrapdoorChannel} and an upper bound given by the feedback capacity \cite{Permuter-2008-TIT-TrapdoorChannelFeedback} are known. Recently, a reinforcement learning approach has been presented in \cite{Permuter-2019-FSCReinforcement} to compute the feedback capacity. For general FSCs, a finite-letter characterization of the capacity in closed form is not known to date; only a general formula based on the inf-information rate has been established in \cite{VerduHan-1994-TIT-GeneralFormulaCapacity}. In this paper, we are interested in the existence of ``\emph{simple}'' capacity expressions and whether or not such capacity expressions for FSCs are algorithmically computable. Both questions are related to each other. For example, a simple capacity expression could be given a single-letter formula with entropic quantities. But it could also be a capacity function which is computable in some sense. The requirement of certain performance functions to be computable is usually implicitly assumed in information theory. Particularly, capacity expressions with entropic quantities in dependence on the communication parameters are usually assumed to be algorithmically computable. 

For the question of algorithmic computability, we use the concept of a \emph{Turing machine} \cite{Turing-1936-ComputableNumbersEntscheidungsproblem,Turing-1937-ComputableNumbersEntscheidungsproblemCorrection,Weihrauch-2000-ComputableAnalysis}, which is a mathematical model of an abstract machine that manipulates symbols on a strip of tape according to certain given rules. It can simulate any given algorithm and therewith provides a simple but very powerful model of computation. Turing machines have no limitations on computational complexity, unlimited computing capacity and storage, and execute programs completely error-free. Accordingly they provide fundamental performance limits for today’s digital computers. Turing machines account for all those problems and tasks that are algorithmically solvable on a classical (i.e., non-quantum) machine. They are further equivalent to the von Neumann-architecture without hardware limitations and the theory of recursive functions, cf. \cite{Godel-1930-VollstandigkeitAxiome,Godel-1934-UndecidablePropositions,Kleene-1952-IntroductionMetamathematics,Minsky-1961-RecursiveUnsolvability,AvigadBrattka-2014-ComputabilityAnalysis}. 

Of particular interest in this work are \emph{computable continuous functions} \cite{PourElRichards-2017-ComputabilityAnalysisPhysics} since such functions can be effectively approximated by computable polynomial sequences which is a very strong requirement on the computability. There are other forms of computability including \emph{Banach-Mazur computability}, which is the weakest from of computability. To this end, Section \ref{sec:computability} introduces the computability framework and studies further properties and insights of computable sequences of computable continuous functions and of Banach-Mazur computable functions. 

Subsequently, this paper studies FSCs which are properly introduced in Section \ref{sec:fsc}. The general question is addressed of whether or not a finite-letter characterization of the capacity exists at all and whether or not the capacity of FSCs is algorithmically computable. In Section \ref{sec:noncomp} it is shown and argued that either the achievability or converse (or both) must result in a non-computable lower or upper bound, respectively. This bound is not even Banach-Mazur computable (and therewith also not Turing computable) and, as a consequence, the capacity is not  Banach-Mazur computable as well. This also means that there exist FSCs for which computable lower and upper bounds can never be tight. Furthermore, it is shown that the capacity of FSCs is not even approximable by computable sequences of computable functions, i.e., it is impossible to approximate the capacity for certain tolerated approximation errors. Note that non-approximability is strictly stronger than non-computability. All these results hold for $|\sX|\geq2$, $|\sY|\geq2$, and $|\sS|\geq2$ and, thus, we consider the general case without restrictions on the cardinalities of the alphabets. This provides a complete picture, since for $|\sS|=1$ the capacity becomes Turing computable and is given by Shannon's single-letter formula. A similar observation with respect to the Turing computability of the capacity of FSCs has been made in \cite{ElkoussPerezGarcia-2018-Nature-Uncomputable}, where it has been shown that the capacity of FSCs is in general not Turing computable if the input and state alphabets $\sX$ and $\sS$ satisfy $|\sX|\geq10$ and $|\sS|\geq62$. This result has been used in \cite{Agarwal-2018-NonExistenceFiniteLetter} to show that for a certain class of entropic formulas, the capacity of time invariant Markov channels cannot be expressed by a finite multi-letter formula. Since this uses \cite{ElkoussPerezGarcia-2018-Nature-Uncomputable} as a ``black box input'', it further only holds for $|\sX|\geq10$ and $|\sS|\geq62$. Our proof relies on completely different techniques than those in \cite{ElkoussPerezGarcia-2018-Nature-Uncomputable} and \cite{Agarwal-2018-NonExistenceFiniteLetter} which further allows us to show that the capacity of FSCs cannot be characterized by a finite-letter entropic expression for input, output, and state alphabets that satisfy $|\sX|\geq2$, $|\sY|\geq2$, and $|\sS|\geq2$. We emphasize that these results hold even for all FSCs with finite input and output alphabets and finite state sets. When the state set is allowed to be countably infinite, the capacity of a computable channel need not be a computable real number anymore. \footnote{\emph{Notation:} $\N$, $\Q$, $\R$, and $\R_c$ are the sets of non-negative integers, rational numbers, real numbers, and computable real numbers; $\sP(\sX)$ and $\sP(\sY|\sX)$ denote the sets of (conditional) probability distributions on $\sY$ (given $\sX$); $H_2(\cdot)$ is the binary entropy function.}

\section{Computability Framework}
\label{sec:computability}

Here, we introduce the computability framework based on Turing machines which provides the needed background. Subsequently, we establish some results on computable sequences which are needed afterwards.

\subsection{Computable Real Numbers and Functions}
\label{sec:computability_numbers}

The concept of computability and computable real numbers was first introduced by Turing in \cite{Turing-1936-ComputableNumbersEntscheidungsproblem} and \cite{Turing-1937-ComputableNumbersEntscheidungsproblemCorrection}. Computable numbers are real numbers that are computable by Turing machines. Since the set of all Turing machines is a countable set, the set of computable real numbers is countable as well. See also the introductory textbook \cite{Weihrauch-2000-ComputableAnalysis} for further details.

A sequence of rational numbers $\{r_n\}_{n\in\N}$ is called a \emph{computable sequence} if there exist recursive functions $a,b,s:\N\rightarrow\N$ with $b(n)\neq0$ for all $n\in\N$ and
\begin{equation}
	\label{eq:computability_comp1}
	r_n= (-1)^{s(n)}\frac{a(n)}{b(n)}, \qquad n\in\N,
\end{equation}
cf.  \cite[Def. 2.1 and 2.2]{Soare-1987-RecursivelyEnumerableSetsDegrees} for a detailed treatment. A real number $x$ is said to be computable if there exists a computable sequence of rational numbers $\{r_n\}_{n\in\N}$ such that
\begin{equation}
	\label{eq:computability_comp2}
	|x-r_n|<2^{-n}
\end{equation}
for all $n\in\N$. This means that the computable real number $x$ is completely characterized by the recursive functions $a,b,s:\N\rightarrow\N$. It has the representation $(a,b,s)$ which we also write as $x\sim (a,b,s)$. It is clear that this representation must not be unique and that there might be other recursive functions $a',b',s':\N\rightarrow\N$ which characterize $x$, i.e., $x\sim (a',b',s')$.

We denote the set of computable real numbers by $\R_c$. Based on this, we define the set of computable probability distributions $\sP_c(\sX)$ as the set of all probability distributions $P_X\in\sP(\sX)$ such that $P_X(x)\in\R_c$ for every $x\in\sX$. The set of all computable conditional probability distributions $\sP_c(\sY|\sX)$ is defined accordingly, i.e., for $P_{Y|X}:\sX\rightarrow\sP(\sY)$ we have $P_{Y|X}(\cdot|x)\in\sP_c(\sY)$ for every $x\in\sX$. This is important since a Turing machine can only operate on computable real numbers.

\begin{definition}
	\label{def:borel}
	A function $f:\R_c\rightarrow\R_c$ is called \emph{Borel computable} if there is an algorithm (or Turing machine) that transforms each given representation $(a,b,s)$ of a computable real number $x$ into a corresponding representation for the computable real number $f(x)$.
\end{definition}

\begin{remark}
	\label{rem:comp}
	From a practical point of view, this can be seen as a minimal requirement for the algorithmic computation of the capacity of a communication system. For this task, an algorithm is needed that takes the communication parameters as inputs to compute the capacity value with a certain precision (e.g. decimal points). In information theory, even for simple problems and questions it cannot be expected that a performance quantity can be exactly numerically computed. For example, for an alphabet $\sX$ of dimension $|\sX|=2$, the entropy $H_2(p)$ of an arbitrary rational probability distribution $p\in\sP(\sX)$ with $p\neq(\frac{1}{2},\frac{1}{2})$ is a transcendental number. 
\end{remark}

Note that Turing's definition of computability conforms to the definition of Borel computability above. In this paper, we will first consider the notion of a \emph{computable continuous function}, cf. for example \cite[Def.~A]{PourElRichards-2017-ComputabilityAnalysisPhysics}. For this, let $\mathbb{I}_c$ denote a computable interval, i.e., $\mathbb{I}_c=[a,b]$ with $a,b\in\R_c$.

\begin{definition}[\cite{PourElRichards-2017-ComputabilityAnalysisPhysics}]
	\label{def:compcont}
	Let $\mathbb{I}_c\subset\R_c$ be a computable interval. A function $f:\mathbb{I}_c\rightarrow\R$ is called  \emph{computable continuous} if:
	\begin{enumerate}
		\item $f$ is \emph{sequentially computable}, i.e., $f$ maps every computable sequence $\{x_n\}_{n\in\N}$ of points $x_n\in\mathbb{I}_c$ into a computable sequence $\{f(x_n)\}_{n\in\N}$ of real numbers,
		\item $f$ is \emph{effectively uniformly continuous}, i.e., there is a recursive function $d:\N\rightarrow\N$ such that for all $x,y\in\mathbb{I}_c$ and all $N\in\N$ with
		\begin{equation*}
			\|x-y\|\leq\frac{1}{d(N)}
		\end{equation*}	
		it holds that
		\begin{equation*}
			|f(x)-f(y)|\leq\frac{1}{2^N}.
		\end{equation*}
	\end{enumerate}
\end{definition}

Computable continuous functions are functions which can be effectively approximated by computable sequence of polynomials $\{P_n\}_{n\in\N}$. Here, every polynomial $P_n$ itself is computable, i.e., its order and coefficients can algorithmically be computed, cf. \cite{PourElRichards-2017-ComputabilityAnalysisPhysics}. Note that the coefficients of these polynomials are usually rational numbers.

There are other forms of computability including \emph{Banach-Mazur computability}, which is the weakest form of computability. In particular, Borel computability and computable continuous functions imply Banach-Mazur computability, but not vice versa. For an overview of the logical relations between different notions of computability we again refer to \cite{AvigadBrattka-2014-ComputabilityAnalysis} and the introductory textbook \cite{Weihrauch-2000-ComputableAnalysis}.

\begin{definition}
	\label{def:banachmazur}
	A function $f:\R_c\rightarrow\R_c$ is called \emph{Banach-Mazur computable} if $f$ maps any given computable sequence $\{x_n\}_{n\in\N}$ of computable real numbers into a computable sequence $\{f(x_n)\}_{n\in\N}$ of computable real numbers.
\end{definition}

If we compare the different notions of computability, we immediate see that any computable continuous function is also Banach-Mazur computable, since Definition \ref{def:banachmazur} is the same as the first condition in Definition \ref{def:compcont}. However, there are infinitely many examples of Banach-Mazur computable functions that are not computable continuous, cf. for example \cite{AvigadBrattka-2014-ComputabilityAnalysis} for a detailed discussion. Such functions do not satisfy the second condition in Definition \ref{def:compcont} and, accordingly, it is not possible to compute the local variations of these functions.

We further need the concepts of a recursive set and a recursively enumerable set as defined e.g. in \cite{Soare-1987-RecursivelyEnumerableSetsDegrees}.

\begin{definition}
	\label{def:recursive}
	A set $\sA\subset\N$ is called \emph{recursive} if there exists a computable function $f$ such that $f(x)=1$ if $x\in\sA$ and $f(x)=0$ if $x\notin\sA$. 
\end{definition}

\begin{definition}
	\label{def:recursiveenumerable}
	A set $\sA\subset\N$ is \emph{recursively enumerable} if there exists a recursive function whose domain is exactly $\sA$.
\end{definition}

We have the following properties; cf. for example \cite{Soare-1987-RecursivelyEnumerableSetsDegrees}
\begin{itemize}
	\item $\sA$ is recursive is equivalent to: $\sA$ is recursively enumerable and $\sA^c$ is recursively enumerable.
	\item There exist recursively enumerable sets $\sA\subset\N$ that are not recursive, i.e., $\sA^c$ is not recursively enumerable. This means there are no computable, i.e., recursive, functions $f:\N\rightarrow\sA^c$ with $[f(\N)]=\sA^c$.
\end{itemize}

\subsection{Computable Sequences of Numbers and Functions}
\label{sec:computability_sequences}

In the following we establish some properties of computable sequences which will be needed subsequently.

\begin{theorem}
	\label{the:theorem1}
	Let $\{x_n^{(1)}\}_{n\in\N}$ and $\{x_n^{(2)}\}_{n\in\N}$ be two computable sequences of computable real numbers with
	\begin{align*}
		x_n^{(1)} \leq x_{n+1}^{(1)}  \quad\text{and}\quad 	x_n^{(2)} \geq x_{n+1}^{(2)}, \quad n\in\N,
	\end{align*}
	and
	\begin{equation*}
		\lim_{n\rightarrow\infty}x_n^{(1)} = \lim_{n\rightarrow\infty}x_n^{(2)} \eqqcolon x_*.
	\end{equation*}
	Then $x_*$ is a computable real number, i.e., $x_*\in\R_c$.
\end{theorem}
\begin{proof}
	If $\{x_n^{(1)}\}_{n\in\N}$ and $\{x_n^{(2)}\}_{n\in\N}$ are computable sequences of rational numbers, then the result can be found in \cite{PourElRichards-2017-ComputabilityAnalysisPhysics}. The proof can be extended to computable real numbers as follows.
	
	Since $\{x_n^{(1)}\}_{n\in\N}$ is a computable sequence of computable real numbers, there is a computable sequence $\{\varphi_n^{(1)}\}_{n\in\N}$ such that for all $N\in\N$ there exists a computable double sequence $\{a_{n,m}^{(1)}\}_{n,m\in\N}$ with
	\begin{equation*}
		\Big|x_n^{(1)}-a_{n,m}^{(1)}\Big| < \frac{1}{2^N} \quad \text{for all }m\geq \varphi_n^{(1)}(N).
	\end{equation*}
	For $m_n = \varphi_n^{(1)}(n)$ we set $a_n^{(1)} = a_{n,m_n}^{(1)}$ so that $\{a_n^{(1)}\}_{n\in\N}$ is a computable sequence of rational numbers and we have
	\begin{equation*}
		x_n^{(1)} > a_n^{(1)}-\frac{1}{2^n}.
	\end{equation*}
	We set $c_n = \max_{1\leq i \leq n}[a_i^{(1)}-\frac{1}{2^i}]$ to obtain the sequence $\{c_n\}_{n\in\N}$ which is a computable sequence of rational numbers with
	\begin{equation*}
		c_n \leq c_{n+1}, \quad n\in\N,
	\end{equation*}
	and 
	\begin{equation*}
		c_n \leq x_n^{(1)} \leq x_*, \quad n\in\N.
	\end{equation*}
	Further, we have
	\begin{align*}
		\big|x_* - c_n\big| &= \big|x_* - a_n^{(1)} + a_n^{(1)} - c_n\big| \\
			&\leq \big|x_* - a_n^{(1)}\big| + \big|a_n^{(1)} - c_n\big| \\
			&= \big|x_* - x_n^{(1)} + x_n^{(1)} - a_n^{(1)}\big| + \big|a_n^{(1)} - c_n\big| \\
			&\leq \big|x_* - x_n^{(1)}\big| + \big|x_n^{(1)} - a_n^{(1)}\big| + \big|a_n^{(1)} - c_n\big| \\
			&\leq \big|x_* - x_n^{(1)}\big| + \frac{1}{2^n} + \frac{1}{2^n}
	\end{align*}
	so that
	\begin{equation*}
		\lim_{n\rightarrow\infty}\big|x_* - c_n\big| = 0,
	\end{equation*}
	i.e., the monotonically increasing computable sequence $\{c_n\}_{n\in\N}$ of rational numbers converges to $x_*$.
	
	In a similar way, based on the computable sequence $\{x_n^{(2)}\}_{n\in\N}$ we can construct a monotonically decreasing computable sequence $\{d_n\}_{n\in\N}$ of rational numbers with
	\begin{equation*}
		\lim_{n\rightarrow\infty}\big|x_* - d_n\big| = 0.
	\end{equation*}
	Now, we can apply the corresponding result from \cite{PourElRichards-2017-ComputabilityAnalysisPhysics} for computable sequences of rational numbers to conclude that $x_*$ must be a computable real number, i.e., $x_*\in\R_c$.
\end{proof}

This allows us to prove the following result.

\begin{theorem}
	\label{the:theorem2}
	Let $\{x_n\}_{n\in\N}$ be a monotonically increasing computable sequence of computable real numbers and let $x_*$ be its limit. If $x_*\in\R_c$, then there exists a recursive function $\varphi:\N\rightarrow\N$ such that for all $N\in\N$ we have for all $n\geq\varphi(N)$
	\begin{equation*}
		\big|x_*-x_n\big| < \frac{1}{2^N}.
	\end{equation*}
\end{theorem}
\begin{proof}
	For computable sequences of rational numbers, the result can be found in \cite{PourElRichards-2017-ComputabilityAnalysisPhysics}. The proof can be extended to computable real numbers as follows. 
	
	We make use of the construction in the proof of Theorem~\ref{the:theorem1} to prove the desired result. Applying this construction to $\{x_n\}_{n\in\N}$ results in a monotonically increasing computable sequence $\{c_n\}_{n\in\N}$ of rational numbers with
	\begin{equation*}
		c_n \leq x_n \leq x_*, \quad n\in\N.
	\end{equation*}
	Since the result holds for monotonically increasing computable sequences of rational numbers, there exists a recursive function $\varphi:\N\rightarrow\N$ such that for all $N\in\N$ we have for all $n\geq\varphi(N)$
	\begin{equation*}
		0 \leq x_*-x_n \leq x_* - c_n < \frac{1}{2^n}
	\end{equation*}
	so that
	\begin{equation*}
		\big|x_*-x_n\big| < \frac{1}{2^N} \quad \text{for all }n\geq\varphi(N).
	\end{equation*}
	Thus, the computable sequence of computable real numbers converges effectively to $x_*$ proving the desired result.
\end{proof}

\begin{remark}
	\label{rem:comp1}
	Note that it is possible to find a computable sequence $\{x_n\}_{n\in\N}$ of rational numbers that converges to a computable real number $x_*\in\R_c$ (which can further be rational), i.e.,
	\begin{equation*}
		\lim_{n\rightarrow\infty}\big|x_*-x_n\big|=0,
	\end{equation*}
	but the convergence is not effective. According to the following Theorem \ref{the:theorem3}, this sequence is then not monotonically increasing or decreasing.
\end{remark}

Next, we establish similar results for computable sequences of computable continuous functions.

\begin{theorem}
	\label{the:theorem3}
	Let $F:[0,1]\rightarrow\R$ be a computable continuous function and $\{F_N\}_{N\in\N}$ be a computable sequence thereof with $F_N(x)\leq F_{N+1}(x)$, $x\in[0,1]$, and
	\begin{equation*}
		\lim_{N\rightarrow\infty}F_N(x) = F(x).
	\end{equation*}
	Then there exists a recursive function $\varphi:\N\rightarrow\N$ such that for all $M\in\N$ we have for all $N\geq\varphi(M)$
	\begin{equation*}
		\big|F(x) - F_N(x)\big| < \frac{1}{2^M}.
	\end{equation*}
\end{theorem}
\begin{proof}
	Let $Q_N(x)=F(x)-F_N(x)$, $x\in[0,1]$. We have $0 \leq Q_{N+1}(x) \leq Q_N(x)$ and $\lim_{N\rightarrow\infty}Q_N(x)=0, \quad x\in[0,1]$. Let $M\in\N$ be arbitrary. There exists an $N_0=N_0(M,x)$ with
	\begin{equation*}
		Q_N(x) < \frac{1}{2^M}\quad\text{for all }N\geq N_0(M,x).
	\end{equation*}
	We define the set
	\begin{equation*}
		\sS_{N,M} = \Big\{x\in[0,1]:Q_N(x)<\frac{1}{2^M}\Big\}
	\end{equation*}
	and observe that $\sS_{N,M} \subset\sS_{N+1,M}$. Now, $\{\sS_{N,M}\}$ is a family of open sets with $[0,1]\subset\bigcup_{N=1}^\infty\sS_{N,M}$. Since $[0,1]$ is a compact set \cite{Rudin87RealComplexAnalysis}, there exists an $N_0(M)$ with $[0,1]\subset\sS_{N_0,M}$ and therewith $Q_{N_0}(x)<\frac{1}{2^M}$ for $N_0$ and also all $N\geq N_0$. Let
	\begin{equation*}
		\max_{x\in[0,1]}Q_N(x) = C_N.
	\end{equation*}
	Since $Q_N$ is a computable continuous function, we always have $C_N\in\R_c$. Further, since $\{Q_N\}_{N\in\N}$ is a computable sequence of computable real numbers, the sequence $\{C_N\}_{N\in\N}$ is also a computable sequence of computable real numbers. For all $N\in\N$ it holds that $C_N\geq C_{N+1}$ and
	\begin{equation*}
		\lim_{N\rightarrow\infty}C_N=0.
	\end{equation*}
	Accordingly, there exists a recursive function $\varphi:\N\rightarrow\N$ such that for all $M\in\N$ we have for all $N\geq\varphi(M)$
	\begin{equation*}
		\big|F(x)-F_N(x)\big| = \big|Q_N(x)\big| < \frac{1}{2^M}
	\end{equation*}
	which proves the desired result.
\end{proof}

Some remarks are in order:
\begin{enumerate}
	\item The result extends to functions on compact spaces.
	\item The result remains true for monotonically decreasing functions.
	\item It is important that $F$ is a computable continuous function. Already for computable sequences of rational numbers with $x_n\leq x_{n+1}$ that converge to a $x_*\notin\R_c$, we do not have effective convergence, see e.g.~\cite{Specker-1949-TJSL-NichtKonstruktivBeweisbar}.
	\item A part of the proof is not effective as we required compactness which is needed to show uniform convergence. This is subsequently used to show the effective convergence of the computable continuous function~$F$.
\end{enumerate}

We can use Theorem \ref{the:theorem3} to show the following result.

\begin{corollary}
	\label{cor:corollary4}
	Let $\{F_N\}_{N\in\N}$ and $\{G_N\}_{N\in\N}$ be computable sequences of computable continuous functions on $[0,1]$ with
	\begin{equation*}
		F_N(x) \leq F_{N+1}(x) \leq G_{N+1}(x) \leq G_N(x)
	\end{equation*}
	and
	\begin{equation*}
		\lim_{N\rightarrow\infty}F_N(x) = \lim_{N\rightarrow\infty}G_N(x) \eqqcolon \Phi(x), \quad x\in[0,1].
	\end{equation*}
	Then $\Phi:[0,1]\rightarrow\R$ is also a computable continuous function and $\{F_N\}_{N\in\N}$ and $\{G_N\}_{N\in\N}$ converge effectively to $\Phi$.
\end{corollary}
\begin{proof}
	We set
	\begin{equation*}
		Q_N(x)=G_N(x)-F_N(x), \quad x\in[0,1],
	\end{equation*}
	and $\{Q_N\}_{N\in\N}$ is a computable sequence of computable continuous functions. For $x\in[0,1]$ we have
	\begin{align*}
		Q_N(x) \geq G_{N+1}(x)-F_N(x) \geq G_{N+1}(x) - F_{N+1}(x)  = Q_{N+1}(x)
	\end{align*}
	and
	\begin{equation*}
		\lim_{N\rightarrow\infty}Q_N(x)=0,\quad x\in[0,1].
	\end{equation*}
	Now, from Theorem \ref{the:theorem3} follows that the computable sequence $\{Q_N\}_{N\in\N}$ of computable continuous functions converges effectively to zero proving the desired result.
\end{proof}

We obtain a similar result for computable sequences of Banach-Mazur computable functions.

\begin{theorem}
	\label{the:bmsequences}
	Let $\{F_N\}_{N\in\N}$ and $\{G_N\}_{N\in\N}$ be computable sequences of functions $F_N:[0,1]\cap\R_c\rightarrow\R_c$ and $G_N:[0,1]\cap\R_c\rightarrow\R_c$, $N\in\N$, with
	\begin{align*}
		F_N(x) &\leq F_{N+1}(x), \quad x\in[0,1]\cap\R_c, \\
		G_N(x) &\geq G_{N+1}(x), \quad x\in[0,1]\cap\R_c,
	\end{align*}
	and
	\begin{equation*}
		\lim_{N\rightarrow\infty}F_N(x) = \lim_{N\rightarrow\infty}G_N(x) \eqqcolon \Phi(x), \quad x\in[0,1]\cap\R_c.
	\end{equation*}
	Then $\Phi:[0,1]\cap\R_c\rightarrow\R$ is also a Banach-Mazur computable function.
\end{theorem}
\begin{proof}
	The function $\Phi:[0,1]\cap\R_c\rightarrow\R$ is well defined. For $x\in[0,1]\cap\R_c$, the function value $\Phi(x)$ is the limit of the monotonically increasing computable sequence $\{F_N(x)\}_{N\in\N}$ of computable real numbers as well as the limit of the monotonically decreasing computable sequence $\{G_N(x)\}_{N\in\N}$ of computable real numbers. Therefore, we have $\Phi(x)\in\R_c$. 
	
	We have to show that for every computable sequence $\{x_n\}_{n\in\N}$ of computable real numbers, the sequence $\{\Phi(x_n)\}_{n\in\N}$ is a computable sequence of computable real numbers as well. Let $y_n=\Phi(x_n)$, $n\in\N$. Similarly as in the proofs of Theorems \ref{the:theorem1} and \ref{the:theorem2}, there exist computable double sequences $\{\overline{y}_{n,N}\}_{n\in\N,N\in\N}$ and $\{\underline{y}_{n,N}\}_{n\in\N,N\in\N}$ of rational numbers with
	\begin{equation*}
		\overline{y}_{n,N}=G_N(x_n) \quad\text{and}\quad \underline{y}_{n,N}=F_N(x_n), 
	\end{equation*}
	which satisfy the following properties: For every $n\in\N$ and $N\in\N$ it holds
	\begin{equation*}
		\overline{y}_{n,N} \geq \overline{y}_{n,N+1} \quad\text{and}\quad \underline{y}_{n,N} \leq \underline{y}_{n,N+1}
	\end{equation*}
	and further
	\begin{equation*}
		\lim_{N\rightarrow\infty}\overline{y}_{n,N}=\lim_{N\rightarrow\infty}\underline{y}_{n,N}=y_n.
	\end{equation*}
	As in the proof of Theorem \ref{the:theorem1}, for $n\in\N$ let for $M\in\N$, $\varphi_n(M)$ be the smallest natural number $N$ such that
	\begin{equation*}
		0 \leq \overline{y}_{n,N} - \underline{y}_{n,N} < \frac{1}{2^M}.
	\end{equation*}
	Then, $\varphi_n$ is a recursive function and $\{\varphi_n\}_{n\in\N}$ is a computable sequence of recursive functions. From the s-m-n-Theorem \cite{Soare-1987-RecursivelyEnumerableSetsDegrees} follows that there exists also a recursive function $\varphi:\N^2\rightarrow\N$ with
	\begin{equation*}
		\varphi(n,M) = \varphi_n(M), \quad (n,M)\in\N^2.
	\end{equation*}
	As in the proof of Theorem \ref{the:theorem2} this implies that for all $n\in\N$ it holds: For all $M\in\N$ we have for all $N\geq\varphi(n,M)$
	\begin{equation*}
		\big|y_n-\underline{y}_{n,N}\big| < \frac{1}{2^M},
	\end{equation*}
	i.e., $\{y_n\}_{n\in\N}$ is a computable sequence of computable real numbers wich completes the proof.
\end{proof}

It is clear that this result also applies to computable sequences of Borel computable functions. Also in this case, the function $\Phi$ must be Banach-Mazur computable. 

In the following, we will use these results and in particular Theorem \ref{the:theorem3}, Theorem \ref{the:bmsequences}, and Corollary \ref{cor:corollary4} to study the computability of the capacity of FSCs.

\section{Finite State Channels}
\label{sec:fsc}

In this section we introduce the concept of finite state channels which are suitable to model discrete channels with memory \cite{Gallager-1968-InformationTheory,BlackwellBreimanThomasian-1958-IndecomposableFSC,Blackwell-1961-InformationTheory}.

\subsection{Basic Definitions}
\label{sec:fsc_def}

Let $\sX$, $\sY$, and $\sS$ be finite input, output, and state sets. FSCs are usually specified by its underlying probability law
\begin{equation}
	\label{eq:probabilitylaw}
	p(y_n,s_n|x_n,s_{n-1}) \in \sP(\sY\times\sS|\sX\times\sS)
\end{equation}
where $y_n\in\sY$ and $s_n\in\sS$ are the output and state of the channel at time instant $n$ whose probability depend on the input $x_n\in\sX$ at time instant $n$ and on the previous state $s_{n-1}\in\sS$ at time instant $n-1$. 

We assume that the output $Y_n$ and state $S_n$ are statistically independent given $x_n$ and $s_{n-1}$ so that \eqref{eq:probabilitylaw} can be written as
\begin{equation}
	\label{eq:probabilitylaw2}
	p(y_n,s_n|x_n,s_{n-1}) = p(y_n|x_n,s_{n-1})q(s_n|x_n,s_{n-1})
\end{equation}
for $p(y_n|x_n,s_{n-1})\in\sP\coloneqq\sP(\sY|\sX\times\sS)$ and $q(s_n|x_n,s_{n-1})\in\sQ\coloneqq\sP(\sS|\sX\times\sS)$. The corresponding sets of computable conditional probabilities are denoted by $\sP_c\coloneqq\sP_c(\sY|\sX\times\sS)$ and $\sQ_c\coloneqq\sP_c(\sS|\sX\times\sS)$, respectively. 

\begin{remark}
	Not that the assumption of independence of $Y_n$ and $S_n$ and its consequence on the probability law as shown in \eqref{eq:probabilitylaw2} will be no loss of generality. In the end, we will show that already the special class \eqref{eq:probabilitylaw2} of FSCs is not Turing computable so that this must be the case for the general class \eqref{eq:probabilitylaw} as well.
\end{remark}

In general, $p^n(y^n|x^n)$ for block length $n$ is undefined for a FSC and we have to consider the general $p^n(y^n,s_n|x^n,s_0)$ which is the probability of the output sequence $y^n$ and a final state $s_n$ at time instant $n$ given an input sequence $x^n$ and an initial state $s_0$. This can be calculated inductively from 
\begin{equation}
	\label{eq:probabilitylaw3}
	\begin{split}
		&p^n(y^n,s_n|x^n,s_0) =\sum_{s_{n-1}\in\sS}p(y_n,s_n|x_n,s_{n-1})p^{n-1}(y^{n-1},s_{n-1}|x^{n-1},s_0),
	\end{split}
\end{equation}
cf. \cite{Gallager-1968-InformationTheory}. Further, by summing over the final state we obtain
\begin{equation}
	\label{eq:probabilitylaw4}
	p^n(y^n|x^n,s_0)=\sum_{s_n\in\sS}p^n(y^n,s_n|x^n,s_0).
\end{equation}

\begin{definition}
	\label{def:code}
	An $(n,M)$-\emph{code} for an FSC consists of an encoder $f:\sM\times\sS\rightarrow\sX^n$ that maps the message $m\in\sM=\{1,...,M\}$ and the initial state $s_0\in\sS$ into the codeword $x^n\in\sX^n$, and a decoder $\varphi:\sY^n\times\sS\rightarrow\sM$ that estimates the transmitted message $\hat{m}\in\sM$ based on the received output $y^n\in\sY^n$ and the initial state $s_0\in\sS$.
\end{definition}

For the initial state $s_0\in\sS$ the average probability of error of such a code based on \eqref{eq:probabilitylaw4} is
\begin{equation*}
	\bar{e}(s_0) = \frac{1}{|\sM|}\sum_{m\in\sM}\sum_{y^n:\varphi(y^n,s_0)\neq m}p(y^n|f(m,s_0),s_0).
\end{equation*}

\begin{definition}
	\label{def:achievable}
	A rate $R>0$ is an \emph{achievable rate} for an FSC if for all $\tau>0$ there exists an $n(\tau)\in\N$ and a sequence of $(n,M)$-codes such that for all $n\geq n(\tau)$ we have $\frac{1}{n}\log M>R-\tau$ and $\bar{e}(s_0)\leq \lambda_n$ for $s_0\in\sS$ with $\lambda_n\rightarrow0$ as $n\rightarrow\infty$. The \emph{capacity} $C$ of an FSC is given by the supremum of all achievable rates $R$.
\end{definition}

The capacity $C$ of an FSC is a function of the communication parameters $p\in\sP$ and $q\in\sQ$, cf. \eqref{eq:probabilitylaw2}, as well as the initial state $s_0\in\sS$. Accordingly, we write $C=C(\{p,q,s_0\})$.

\subsection{General Capacity Formulas}
\label{sec:fsc_capacity}

We will study the computability of the capacity function $C$ in dependence on the communication parameters $\{p,q,s_0\}$. These will be the inputs to the corresponding Turing machine. For this purpose, we need a corresponding expression for $C(\{p,q,s_0\})$ as for example the general formula provided by Verd\'u and Han in \cite{VerduHan-1994-TIT-GeneralFormulaCapacity}. For the FSC as defined above, the capacity can be expressed in a multi-letter form as
\begin{equation}
	\label{eq:hanverdu}
	C(\{p,q,s_0\}) = \lim_{n\rightarrow\infty}\sup_{X^n}\frac{1}{n}I(X^n;Y^n|s_0)
\end{equation}
according to the underlying probability law \eqref{eq:probabilitylaw3}-\eqref{eq:probabilitylaw4}. This has been shown to be valid for information stable channels \cite{Dobrushin-1963-AMS-GeneralFormulation}, but does not hold in full generality. Moreover, this expression cannot be computed immediately as it is the limit of a sequence of optimization problems. Furthermore, it is not even clear if $C(\{p,q,s_0\})$ is a computable real number for computable $p$ and $q$. Another formula for the capacity based on the inf-information rate has been established in \cite{VerduHan-1994-TIT-GeneralFormulaCapacity} 
\begin{equation}
	\label{eq:hanverdu2}
	C(\{p,q,s_0\}) = \sup_{\vec{X}}\underline{\vec{I}}(\vec{X};\vec{Y})
\end{equation}
where $\underline{\vec{I}}(\vec{X};\vec{Y})$ is the inf-information rate as defined in \cite{HanVerdu-1993-TIT-ApproximationOutputStatistics}. In general, this expression cannot be evaluated easily. 

For the special class of so-called indecomposable channels, there exists a simple capacity expression for $C(\{p,q,s_0\})$. This is discussed next.

\subsection{Indecomposable Channels}
\label{sec:fsc_indecomp}

The class of IFSCs goes back to \cite{BlackwellBreimanThomasian-1958-IndecomposableFSC} and refers to those FSCs for which the effect of the initial state vanishes with time. For the precise definition of this, we follow \cite[Sec. 4]{Gallager-1968-InformationTheory} and set $q^n(s_n|x^n,s_0) = \sum_{y^n\in\sY^n}p^n(y^n,s_n|x^n,s_0)$.

\begin{definition}
	\label{def:indecomp}
	An FSC is called \emph{indecomposable} if for every $\epsilon>0$ there exists an $n_0\in\N$ such that for all $n\geq n_0$ we have $\big|q^n(s_n|x^n,s_0) - q^n(s_n|x^n,s_0')\big| \leq \epsilon$ for all $s_n\in\sS$, $x^n\in\sX^n$, $s_0\in\sS$, and $s_0'\in\sS$.
\end{definition}

For the capacity of IFSCs, we need the functions: 
\begin{subequations}\label{eq:clncun}
\begin{align}
	\Cl_n(\{p,q\}) &= \frac{1}{n}\max_{X^n}\min_{s_0}I(X^n;Y^n|s_0) \label{eq:cln}\\
	\Cu_n(\{p,q\}) &= \frac{1}{n}\max_{X^n}\max_{s_0}I(X^n;Y^n|s_0). \label{eq:cun}
\end{align}
\end{subequations}

\begin{remark}
	Note that for fixed $n\in\N$ and computable parameters $\{p,q,s_0\}\in\sP_c\times\sQ_c\times\sS$, the functions $\Cl_n$ and $\Cu_n$ in \eqref{eq:clncun} are computable functions, i.e., we have $\Cl_n:\sP_c\times\sQ_c\times\sS\rightarrow\R_c$ and $\Cu_n:\sP_c\times\sQ_c\times\sS\rightarrow\R_c$.
\end{remark}

The sequences $\{\Cl_n\}_{n=1}^\infty$ and $\{\Cu_n\}_{n=1}^\infty$ for $\{p,q,s_0\}\in\sP\times\sQ\times\sS$ converge and we have
\begin{subequations}\label{eq:clcu}
	\begin{align}
	\Cl(\{p,q\}) = \lim_{n\rightarrow\infty} \Cl_n(\{p,q\}) \label{eq:cl}\\
	\Cu(\{p,q\}) = \lim_{n\rightarrow\infty} \Cu_n(\{p,q\}) \label{eq:cu}
	\end{align}
\end{subequations}
which are also called \emph{lower capacity} and \emph{upper capacity}, respectively. If the FSC is indecomposable, lower and upper capacities coincide and are equal to the capacity, i.e., $\Cl(\{p,q\}) = \Cu(\{p,q\}) = C(\{p,q,s_0\})$.

\subsection{Main Problem Formulation}
\label{sec:fsc_problem}

For fixed alphabets $\sX$, $\sY$, and $\sS$, the capacity $C$ of an FSC is a function of the underlying system parameters $\{p,q,s_0\}$. The previous discussion leads to the following questions of interest:
\vspace*{0.25\baselineskip}

\begin{tcolorbox}[colback=white,boxrule=0.125ex]
	{\bf Question 1:} Is the capacity $C(\{p,q,s_0\})$ for fixed initial state $s_0\in\sS$ a computable continuous function on $\sP$ and $\sQ$?
\end{tcolorbox}

Here, we allow arbitrary $p\in\sP$ and $q\in\sQ$ as inputs for the capacity function $C$. From Corollary \ref{cor:corollary4} we alrady see that this question is naturally connected to the question whether or not it is possible to find computable continuous lower and upper bounds on the capacity. Lower bounds originate from actual coding schemes and general achievability results, while the upper bounds are established via converse arguments. 

From a practical point of view, such lower and upper bounds should be computable to enable a numerical evaluation on digital computers. Therefore, it is reasonable to study the capacity $C$ as a function on computable inputs $p\in\sP_c$ and $q\in\sQ_c$. This leads to following question:
\vspace*{0.5\baselineskip}

\begin{tcolorbox}[colback=white,boxrule=0.125ex]
	{\bf Question 2:} Is the capacity $C(\{p,q,s_0\})$ for fixed initial state $s_0\in\sS$ a Borel computable function (and therewith Turing computable)?
\end{tcolorbox}
\vspace*{0.25\baselineskip}

While Borel computability is a strong notion of computability, Banach-Mazur computability is considered to be the weakest form of computability and it is of interest to pose a similar question for this notion as follows:
\vspace*{0.5\baselineskip}

\begin{tcolorbox}[colback=white,boxrule=0.125ex]
	{\bf Question 3:} Is the capacity $C(\{p,q,s_0\})$ for fixed initial state $s_0\in\sS$ a Banach-Mazur computable function?
\end{tcolorbox}
\vspace*{0.25\baselineskip}

As for the first question, the lower and upper bounds on the capacity in Questions 2 and 3 should be algorithmically computable. 
\vspace*{0.5\baselineskip}

\begin{tcolorbox}[colback=white,boxrule=0.125ex]
	{\bf Question 4:} Is the capacity $C(\{p,q,s_0\})$ for fixed initial state $s_0\in\sS$ approximately Turing computable?
\end{tcolorbox}

\begin{remark}
	In the following, we will provide negative answers to Questions 1-3. As the capacity is shown to be non-computable, Question 4 about whether or not the capacity is at least approximately computable becomes particularly relevant. To make sure that this question is not trivial, the tolerated approximation error should not be too large. Also for Question 4 we will provide a negative answer for certain approximation errors.
\end{remark}

\section{Computability Analysis of the FSC Capacity}
\label{sec:noncomp}

In this section, we show that the capacity function $C$ is not Banach-Mazur computable and therewith also not Borel and Turing computable. Subsequently, we discuss the case when the capacity of an FSC becomes a computable real number.

\subsection{Non-Banach-Mazur Computability}
\label{sec:noncomp_bm}

In general, the capacity of an FSC is given by \eqref{eq:hanverdu2} and for every $n\in\N$, the inf-information rate expression $\sup_{\vec{X}}\underline{\vec{I}}(\vec{X};\vec{Y})$ is indeed Turing computable (we omit the details due to space constraints). However, in the end the capacity in \eqref{eq:hanverdu2} is given by the limit of for $n\rightarrow\infty$ and the convergence of this limit need not be effective and uniform on $\{p,q,s_0\}\in\sPc\times\sQc\times\sS$, i.e., for a given $\epsilon\in\Q$, e.g., $\epsilon^*=\frac{1}{\mu}$ with $\mu\in\N$, we cannot algorithmically compute when $|f_n(p,q,s_0)-C(\{p,q,s_0\})|<\epsilon$ is satisfied.

\begin{theorem}
	\label{the:comp}
	For all $|\sX|\geq2$, $|\sY|\geq2$, and $|\sS|\geq2$, the capacity function $C(\{p,q,s_0\}):\sPc\times\sQc\times\sS\rightarrow\R$ of the FSC with parameters $\{p,q,s_0\}$ is not Banach-Mazur computable.
\end{theorem}
\begin{proof}
	We first prove the result for $|\sX|=|\sY|=|\sS|=2$ and subsequently outline how it extends to the general case.
	
	If the finite state channel $\{p,q,s_0\}$, $s_0\in\sS=\{0,1\}$, is indecomposable, then the effect of the initial state vanishes and we have
	\begin{equation*}
		C(\{p,q,0\}) = C(\{p,q,1\}) = \Cu(\{p,q\}) = \Cl(\{p,q\})
	\end{equation*}
	and further
	\begin{align*}
		\Cl(\{p,q\}) = \min_{s_0\in\{0,1\}}C(\{p,q,s_0\}) \leq \max_{s_0\in\{0,1\}}C(\{p,q,s_0\}) = \Cu(\{p,q\}) .
	\end{align*}
	Next we consider the channel 
	\begin{equation}
		p(y_n|x_n,0) = \begin{pmatrix}
			1 & 0 \\
			0 & 1
		\end{pmatrix}\!,\quad
		p(y_n|x_n,1) = \begin{pmatrix}
			1\!-\!\epsilon & \epsilon \\
			\epsilon & 1\!-\!\epsilon
		\end{pmatrix}
		\label{eq:p}
	\end{equation}
	for some $0<\epsilon<1/2$, i.e., for state $s_{n-1}=0$ the channel is noiseless, while for $s_{n-1}=1$ it is noisy. Further, we consider the state distribution
	\begin{equation}
		\hat{q}(s_n|x_n,0) = \begin{pmatrix}
			1  \\
			0 
		\end{pmatrix},\quad
		\hat{q}(s_n|x_n,1) = \begin{pmatrix}
			0  \\
			1
		\end{pmatrix}
		\label{eq:hatq}
	\end{equation}
	to be independent of $x_n\in\sX$ so that for $s_n\in\sS$ and $s_{n-1}\in\sS$ arbitrary we have
	\begin{equation}
		\label{eq:fadingstate}
		\hat{q}(s_n|x_n,s_{n-1}) = \hat{q}(s_n|s_{n-1}).
	\end{equation}
	Note that $p$ and $\hat{q}$ as defined above are computable, i.e., we have $p\in\sPc\coloneqq\sP_c(\sY|\sX\times\sS)$ and $\hat{q}\in\sQc\coloneqq\sP_c(\sS|\sX\times\sS)$. In what follows, we consider the finite state channel $\{p,\hat{q},s_0\}$, $s_0\in\{0,1\}$, as defined above.
	
	We observe that $\{p,\hat{q},0\}$ is given by a simple discrete memoryless channel (DMC) $p(y|x,0)$, $x\in\sX$, $y\in\sY$, since the state is always $s_n=0$ due to the definition of $\hat{q}$, cf. \eqref{eq:hatq}. Accordingly, the capacity is $C(\{p,\hat{q},0\})=1$ in this case, since the alphabets are binary and the channel is noiseless. 
	
	We further observe that $\{p,\hat{q},1\}$ corresponds to the DMC $p(y|x,1)$, $x\in\sX$, $y\in\sY$, i.e., it is a binary symmetric channel (BSC). The optimal input distribution is known to be the uniform distribution and the capacity in this case is then $C(\{p,\hat{q},1\})=C_{\text{BSC}}(\epsilon)=1-H_2(\epsilon)<1$.
	
	Next, we show that both functions $C(\{p,q,0\})$ and $C(\{p,q,1\})$ cannot be simultaneously Banach-Mazur computable. For this purpose, we take an arbitrary recursively enumerable, but not recursive, set $\sA\subset\N$. Let $\fT_\sA$ be a Turing machine that stops if and only if for input $n$ we have $n\in\sA$. Otherwise, $\fT_\sA$ runs forever. Such a Turing machine can easily be found as argued next: Let $\varphi_\sA:\N\rightarrow\N$ be a recursive function that lists all elements of the set $\sA$ and for which $\varphi_\sA:\N\rightarrow\sA$ is a unique function.
	
	Let $n\in\N$ be arbitrary. The Turing machine $\fT_\sA$ with input $n$ is defined as follows: We start with $l=1$ and compute $\varphi_\sA(1)$. If $n=\varphi_{\sA}(1)$, then the Turing machine stops. In the other case, the Turing machine computes $\varphi_\sA(2)$. Similarly, if $n=\varphi_{\sA}(2)$, then the Turing machine stops and otherwise, it continues computing the next element. It is clear that this Turing machine stops if and only if $n\in\sA$. 
	
	Assume that both functions $C(\{p,q,0\})$ and $C(\{p,q,1\})$ are Banach-Mazur computable. For $\lambda\in[0,\frac{1}{2}]\cap\R_c$ we consider
	\begin{align*}
		q_\lambda(s_n|x_n,0) = \begin{pmatrix}
			1-\lambda \\ \lambda
		\end{pmatrix}
		\quad\text{and}\quad		
		q_\lambda(s_n|x_n,1) = \begin{pmatrix}
			\lambda \\ 1-\lambda
		\end{pmatrix}.
	\end{align*}
	Of course, for $\lambda\in[0,\frac{1}{2}]\cap\R_c$, $q_\lambda(s_n|x_n,0)$ and $q_\lambda(s_n|x_n,1)$ are computable probability distributions. Let
	\begin{align*}
		q_0(s_n|x_n,0) = \hat{q}(s_n|x_n,0), \quad s_n\in\sS, x_n\in\sX, \\
		q_0(s_n|x_n,1) = \hat{q}(s_n|x_n,1), \quad s_n\in\sS, x_n\in\sX.
	\end{align*}
	We have
	\begin{equation*}
		C(\{p,q_0,1\})-C(\{p,q_0,0\}) = 1 - (1-H_2(\epsilon)) = H_2(\epsilon) > 0.
	\end{equation*}
	For $0<\lambda\leq\frac{1}{2}$ the FSC $\{p,q_\lambda,s_0\}$ $s_0\in\sS$ is indecomposable and therewith we have
	\begin{equation*}
		C(\{p,q_\lambda,0\})=C(\{p,q_\lambda,1\}).
	\end{equation*}
	Now, for every $n\in\N$ and $m\in\N$ let
	\begin{align*}
		\lambda_{n,m} = \begin{cases}
			\frac{1}{2^l} &\quad \fT_\sA \text{ stops for input } n \text{ after } l\leq m \text{ steps} \\
			\frac{1}{2^m} &\quad \fT_\sA \text{ does not stop for input } n \text{ after } m \text{ steps}. 
		\end{cases}
	\end{align*}
	Then the sequence $\{\lambda_{n,m}\}_{n,m\in\N}$ is a computable double sequence of rational numbers. For arbitrary $n\in\N$ and arbitrary $m,m_1\in\N$, $m\geq m_1$, it holds
	\begin{equation}
		\label{eq:bmB}
		\big|\lambda_{n,m} - \lambda_{n,m_1}\big| = \big|\lambda_{n,m_1}-\lambda_{n,m}\big| = \lambda_{n,m_1}-\lambda_{n,m} < \frac{1}{2^{m_1}}
	\end{equation} 
	since if the Turing machine $\fT_\sA$ has stopped for input $n$ for $l\leq m_1$, then $\lambda_{n,m_1}=\lambda_{n,m}$ and \eqref{eq:bmB} is trivially satisfied. If the Turing machine $\fT_\sA$ has not stopped for input $n$ after $m_1$ iterations, then $\lambda_{n,m_1}=\lambda_{n,m}=\frac{1}{2^{m_1}}-\lambda_{n,m}<\frac{1}{2^{m_1}}$, since $\lambda_{n,m}>0$ for all $n\in\N$, so that \eqref{eq:bmB} is satisfied as well. Accordingly, we observe that $\{\lambda_{n,m}\}_{m\in\N}$ is a sequence that converges effectively and there exists one $\lambda_n^*\in\R_c$ with
	\begin{equation*}
		\lim_{m\rightarrow\infty}\big|\lambda_n^*-\lambda_{n,m}\big|=0.
	\end{equation*}
	Furthermore, since $\{\lambda_{n,m}\}_{n,m\in\N}$ is computable double sequence, the sequence $\{\lambda_n^*\}_{n\in\N}$ is a computable sequence of computable real numbers. It further holds $\lambda_n^*\geq0$ with equality if and only if the Turing machine $\fT_\sA$ does not stop for input $n$. 
	
	Since $C(\{p,q,0\})$ and $C(\{p,q,1\})$ are assumed to be Banach-Mazur computable functions, the difference $\Phi(\{p,q\})=C(\{p,q,1\})-C(\{p,q,0\})$ is a Banach-Mazur computable function as well. Then, the sequence $\{\mu_n\}_{n\in\N}$ with 
	\begin{equation*}
		\mu_n = \Phi(\{p,q_{\lambda_n^*}\}), \quad n\in\N,
	\end{equation*}
	is a computable sequence of computable real numbers. With this, we find a computable double sequence $\{\nu_{n,m}\}_{n,m\in\N}$ of rational numbers with
	\begin{equation*}
		\big|\mu_n-\nu_{n,m}\big| < \frac{1}{2^m}.
	\end{equation*}
	For every $n$, we can consider the following Turing machine $\fT_*$: For input $n$, we set $m=1$ and check if
	\begin{equation*}
		\nu_{n,1} > \frac{1}{2}
	\end{equation*}
	is satisfied. If this is true, the Turing machine stops. Otherwise, we set $m=2$ and check if
	\begin{equation*}
	\nu_{n,2} > \frac{1}{4}
	\end{equation*}
	is satisfied. If this is true, the Turing machine stops. Otherwise, it continues as described. Next, we show that this Turing machine $\fT_*$ stops for input $n$ if and only if $\mu_n>0$.
	
	``$\Leftarrow$'' If $\mu_n>0$, then there exists an $m_0$ with
	\begin{equation*}
		\frac{1}{2^{m_0}} < \frac{\mu_n}{2}
	\end{equation*}
	so that
	\begin{align*}
		\mu_n &= \mu_n - \nu_{n,m_0} + \nu_{n,m_0} \leq \big|\mu_n - \nu_{n,m_0}\big| + \nu_{n,m_0} \\
			&< \frac{1}{2^{m_0}} + \nu_{n,m_0} < \frac{\mu_n}{2} + \nu_{n,m_0}, 
	\end{align*}
	i.e., the Turing machine $\fT_*$ stops for input $n$ within $m_0$ steps.
	
	``$\Rightarrow$'' It holds $\nu_{n,\hat{m}}>\frac{1}{2^{\hat{m}}}$ for a certain $\hat{m}$. Then,
	\begin{align*}
		\frac{1}{2^{\hat{m}}} &< \nu_{n,\hat{m}} = \nu_{n,\hat{m}} - \mu_n + \mu_n \\
			&\leq \big|\nu_{n,\hat{m}} - \mu_n\big| + \mu_n < \frac{1}{2^{\hat{m}}} + \mu_n
	\end{align*}
	so that $\mu_n>0$ is true. 
	
	Next, for input $n\in\N$, we define the Turing machine $\fT_\sS$ as follows: We run both previous Turing machines $\fT_\sA$ and $\fT_*$ in parallel for input $n$, where each Turing machine operates step by step as discussed above. We have already shown that $\fT_\sA$ stops for input $n$ if and only if $n\in\sA$. Further, we have shown that $\fT_*$ stops for input $n$ if and only if $\mu_n>0$. This is true if and only if the Turing machine $\fT_\sA$ does not stop for input $n$, i.e., whenever $n\in\sA^c$. As a consequence, one of these Turing machines must always stop for an input $n$. We set
	\begin{align*}
		\fT_\sS(n) = \begin{cases}
			\big\{n\in\sA\big\} &\text{ if } \fT_\sA \text{ stops for input } n \\
			\big\{n\in\sA^c\big\} &\text{ if } \fT_* \text{ stops for input } n.
			\end{cases}
	\end{align*}
	With this, we have shown that $\sA$ is a recursive set. But this is a contradiction so that the assumption that both functions $C(\{p,q,0\})$ and $C(\{p,q,1\})$ are Banach-Mazur computable is wrong. This completes the proof.
	
	For $|\sX|\geq2$, $|\sY|\geq2$, and $|\sS|\geq2$ arbitrary, we take the sequences of parameters $\{p,\hat{q}\}$ and $\{p,q_k\}$ as above and extend them as follows: We set $p(y|x_n,s_{n-1})=0$ for $y\in\sY\backslash\{0,1\}$, $x_n\in\sX$, $s_{n-1}\in\sS$ and also for $x_n\in\sX\backslash\{0,1\}$, $s_{n-1}\in\sS\backslash\{0,1\}$ to preserve the above constructed behavior. We do the same for $\hat{q}$ and $q_k$. We observe that we still have $p\in\sPc$ and $\hat{q},q_k\in\sQc$. With this and the previous arguments we can conclude on the same result.
\end{proof}

\begin{remark}
	This result and implications thereof can further be strengthened for countably infinite state sets. In particular, for computable compound channels with countably infinite state sets, the capacity need not be a computable real number in general, cf. also \cite{BocheSchaeferPoor-2020-TSP-CommunicationChannelUncertainty}.
\end{remark}

\begin{remark}
	The techniques used to prove Theorem \ref{the:comp} can be extended to various channel models and operational (communication) tasks in information theory. For example, the problem of secret key generation with rate-limited public discussion has been studied in \cite{BocheSchaeferBaurPoor-2019-TSP-ComputabilitySKGAuthentication} and the problem of identification with feedback in \cite{BocheSchaeferPoor-2020-TIT-IDF}.
\end{remark}

\begin{remark}
	The proof of Theorem \ref{the:comp} provides additional deeper insights. This has been developed in detail in \cite{BocheSchaeferPoor-2020-TIT-IDF} for the identification with feedback capacity. By modifying the proof above, one is able to show the following: It is possible to connect the algorithmic computation of the capacity to hard problems in pure mathematics such as Goldbach's Conjecture and the Riemann Hypothesis. To this end, it is possible to find an oracle Turing machine with the following properties: Given finitely many values of the capacity function of the given computable channel, the oracle Turing machine that gets the capacity value of certain computable FSCs as oracle can immediately prove or disprove Goldbach's Conjecture and the Riemann Hypothesis.
\end{remark}

\begin{remark}
	It is not clear if similar results hold for the capacity of time-continuous channels as in \cite{Kailath-1959-TechRep-SamplingLinearTimeVariantFilter}. Accordingly, it is not clear if the technique presented above is applicable in this case at all. A more detailed discussion on this is given in Section \ref{sec:discussion}.
\end{remark}

In the construction of the proof of Theorem \ref{the:comp} above, we assume the special case in which the current state $s_n$ does not depend on the current input $x_n$ but only on the previous state $s_{n-1}$. This is the special class of \emph{finite fading channels (FFCs)} that naturally applies to wireless communications where the fading state of the channel is independent of the transmitted signal. We immediately obtain the following corollary.

\begin{corollary}
	\label{cor:ffc}
	For all $|\sX|\geq2$, $|\sY|\geq2$, and $|\sS|\geq2$, the capacity function $C(\{p,q,s_0\}):\sPc\times\sQc\times\sS\rightarrow\R$ of the FFC with parameters $\{p,q,s_0\}$ is not Banach-Mazur computable. 
\end{corollary}

We see that, in general, the capacity of an FSC is not Banach-Mazur and therewith also not Turing computable. However, for special cases of FSCs the capacity becomes Turing computable as e.g. the zero-error capacity \cite{AhlswedeKaspi-1987-TIT-PermutingChannels} or the feedback capacity  \cite{Permuter-2008-TIT-TrapdoorChannelFeedback} of the trapdoor channel; but in general, there is no algorithm that can compute the capacity as a function of the parameters $\{p,q,s_0\}$.

\begin{remark}
	\label{rem:compnumber}
	Banach-Mazur computability requires the function to operate on computable reals, cf. Definition~\ref{def:banachmazur}. In Theorem~\ref{the:comp} we have shown that $C(\{\cdot,\cdot,s_0\})$ is not Banach-Mazur computable, but this does not imply that the function $C(\{\cdot,\cdot,s_0\})$ itself is not a mapping from computable probability distributions to computable reals, i.e., 
	\begin{equation}
		\label{eq:banachmazur}
		C(\{\cdot,\cdot,s_0\}):\sPc\times\sQc\rightarrow\R_c \quad\text{ for all } s_0\in\sS.
	\end{equation}
	The problem in showing this, is the following: Although the capacity expression \eqref{eq:hanverdu2} is a multi-letter formula which converges, the speed of convergence does not need to be effective, i.e., it may not be representable by an effectively computable function. And indeed, it is not clear whether or not the convergence of \eqref{eq:hanverdu2} is effective. 
\end{remark}

Next, we study the existence of computable tight lower and upper bounds on the capacity function. First, we study such bounds which are computable continuous functions on the parameters $\{p,q\}$. As lower and upper bounds should be numerically evaluable, this is a very reasonable requirement, cf. also Remark \ref{rem:comp}.

\begin{theorem}
	\label{the:fsctheorem5}
	For $|\sX|\geq2$, $|\sY|\geq2$, and $|\sS|\geq2$ arbitrary but fixed,  there exists an $s_0\in\sS$ such that the following holds: There exists no computable sequences $\{F_N\}_{N\in\N}$ and $\{G_N\}_{N\in\N}$ of computable continuous functions with
	\begin{enumerate}
		\item $F_N:\sP\times\sQ\rightarrow\R$ and $G_N:\sP\times\sQ\rightarrow\R$, $N\in\N$,
		\item $F_N(p,q)\leq C(\{p,q,s_0\})$, $p\in\sP$, $q\in\sQ$, $N\in\N$, and $\lim_{N\rightarrow\infty}F_N(p,q)=C(\{p,q,s_0\})$ for all $p\in\sP$, $q\in\sQ$,
		\item $C(\{p,q,s_0\})\leq G_N(p,q)$, $p\in\sP$, $q\in\sQ$, $N\in\N$, and $\lim_{N\rightarrow\infty}G_N(p,q)=C(\{p,q,s_0\})$ for all $p\in\sP$, $q\in\sQ$.
	\end{enumerate}
\end{theorem}
\begin{proof}
	The result follows immediately from Corollary \ref{cor:corollary4}. If such sequences $\{F_N\}_{N\in\N}$ and $\{G_N\}_{N\in\N}$ would exist, then $C$ would be a computable continuous function which is a contradiction, since $C$ is for a certain $s_0\in\sS$ not Banach-Mazur computable. 
\end{proof}

This result shows that an approximation of $C$ by computable continuous functions is not possible. From this, we can immediately conclude the following.

\begin{corollary}
	\label{cor:fsccor2}
	For all computable sequences $\{F_N\}_{N\in\N}$ and $\{G_N\}_{N\in\N}$ of computable continuous functions for which there exists an $s_0\in\sS$ such that for $N\in\N$ it holds that
	\begin{equation*}
		F_N(p,q) \leq C(\{p,q,s_0\})
	\end{equation*}
	for all $p\in\sP$ and $q\in\sQ$, and for $N\in\N$ it holds that
	\begin{equation*}
		C(\{p,q,s_0\}) \leq G_N(p,q)
	\end{equation*}
	for all $p\in\sP$ and $q\in\sQ$, there must exist a $(p_*,q_*)\in\sP\times\sQ$ such that
	\begin{equation}
	\begin{split}
		0 < \max\Big\{&\limsup_{N\rightarrow\infty}\big|C(\{p_*,q_*,s_0\})-F_N(p_*,q_*)\big|,\\
			&\qquad \limsup_{N\rightarrow\infty}\big|C(\{p_*,q_*,s_0\})-G_N(p_*,q_*)\big|\Big\}.
		\label{eq:fsccor}
	\end{split}
	\end{equation}
\end{corollary}
\begin{proof}
	These statements follow immediately from Theorem \ref{the:fsctheorem5}, since if \eqref{eq:fsccor} would be zero for all $(p,q)\in\sP\times\sQ$, then this would imply that $C$ is a computable function.
\end{proof}

As a consequence from this result we can conclude that for the capacity of general FSCs, there is either no computable achievability or no computable converse (or both are non-computable).

The functions $\{F_N\}$ can be interpreted as lower bounds for achievable rates and the capacity respectively. Of course, such bounds should be effectively computable so that they can be numerically evaluated. These bounds should improve with increasing $N\in\N$, i.e., $F_N(p,q)\leq F_{N+1}(p,q)$, $(p,q)\in\sP\times\sQ$, and further should be asymptotically tight, i.e., for $N\rightarrow\infty$ the sequence $\{F_N\}_{N\in\N}$ should converge pointwise to $C(\{p,q,s_0\})$.

Accordingly, the functions $\{G_N\}$ can be seen as upper bounds on the achievable rates and the capacity respectively. Similarly, it is required that these bounds are effectively computable and further $C(\{p,q,s_0\})\leq G_{N+1}(p,q)\leq G_N(p,q)$, $(p,q)\in\sP\times\sQ$, i.e., the bounds should improve with increasing $N\in\N$.

However, Corollary \ref{cor:fsccor2} shows that we cannot find such functions $\{F_N\}$ and $\{G_N\}$. Accordingly, it is impossible that both achievability and converse are effectively computable at the same time. As a consequence, one of these must be non-computable so that we cannot find a entropic characterization for the capacity. This also means that there exist computable FSCs for which computable lower and upper bounds can never be simultaneously be tight.

\begin{remark}
	Finally, we note that the results of Theorem \ref{the:fsctheorem5} and Corollary \ref{cor:fsccor2} remain true if the requirement of $\{F_N\}_{N\in\N}$ and $\{G_N\}_{N\in\N}$ being computable sequences of computable continuous functions is weakened to computable sequences of Banach-Mazur computable sequences.
\end{remark}

Note that Corollary \ref{cor:fsccor2} further provides a negative answer to Question~4. In particular, the proof of Theorem \ref{the:fsctheorem5} yields lower bounds for the error, for which the capacity cannot be approximated. Note that the statement of non-approximability is strictly stronger than the statement of non-Turing-computability. Indeed, with the results in \cite{BocheSchaeferPoor-2019-ITW-NonIID} it is possible to show that there are channels whose capacity is not Turing computable but are approximable for any given approximation error.

\subsection{Capacity being a Computable Real Number}
\label{sec:noncomp_number}

Next, we further study the behavior of the capacity function \eqref{eq:banachmazur} in more detail and address the question if the capacity value itself is a computable real number, cf. also Remark \ref{rem:compnumber}. The following Theorem \ref{the:cu} provides a result for a large class of computable FSCs.

\begin{theorem}
	\label{the:cu}
	For every computable FSC $\{p,q,s_0\}$, $s_0\in\sS$, that satisfies $\Cu(\{p,q\})=\Cl(\{p,q\})$, we have $C(\{p,q,s_0\})\in\R_c$ for all $s_0\in\sS$, i.e., the capacity is a computable real number.
\end{theorem}
\begin{proof}
	We make use of the following properties. Let $p\in\sP$ and $q\in\sQ$ be arbitrary. Then 
	\begin{align*}
		\Cu(\{p,q\}) &= \inf_{n\in\N}\Big(\Cu_n(\{p,q\})+\frac{\log|\sS|}{n}\Big), \\
		\Cl(\{p,q\}) &= \sup_{n\in\N}\Big(\Cl_n(\{p,q\})-\frac{\log|\sS|}{n}\Big),
	\end{align*}
	see \cite[Theorem 4.6.1]{Gallager-1968-InformationTheory}.
	
	For every $n\in\N$ and $\{p,q\}\in\sPc\times\sQc$, $\Cu_n(\{p,q\})$ is a computable number. Accordingly, $\{\Cu_n(\{p,q\})\}_{n=1}^\infty$ is a computable sequence of computable reals. We define
	\begin{equation}
		\label{eq:cum}
		\Cu(M;\{p,q\}) \coloneqq \min_{1\leq n\leq 2^M}\Big(\Cu_n(\{p,q\})+\frac{\log|\sS|}{n}\Big).
	\end{equation}
	We see that $\Cu(M;\{p,q\})$ is a computable real for $M\in\N$ and the corresponding sequence $\{\Cu(M;\{p,q\})\}_{n=1}^\infty$ is a computable sequence of computable reals. We have $\Cu(M;\{p,q\}) \geq \Cu(M+1;\{p,q\})$ for $M\in\N$, i.e., the sequence is monotonically decreasing and it holds $\lim_{M\rightarrow\infty} \Cu(M;\{p,q\}) = \Cu(\{p,q\})$. We further set
	\begin{equation}
		\label{eq:clm}
		\Cl(M;\{p,q\}) \coloneqq \max_{1\leq n\leq 2^M}\Big(\Cl_n(\{p,q\})-\frac{\log|\sS|}{n}\Big)
	\end{equation}
	and similarly obtain $\Cl(M+1;\{p,q\}) \geq \Cl(M;\{p,q\})$ for $M\in\N$. It holds $\lim_{M\rightarrow\infty} \Cl(M;\{p,q\}) = \Cl(\{p,q\})$. By assumption we further have for all $s_0\in\sS$, $\Cl(\{p,q\}) = C(\{p,q,s_0\}) = \Cu(\{p,q\})$.
	
	Next, we consider the function
	\begin{equation*}
		g_M(\{p,q\}) \coloneqq \Cu(M;\{p,q\}) - \Cl(M;\{p,q\}). 
	\end{equation*}
	Due to the monotonicity of both sequences, we have $0 \leq g_{M+1}(\{p,q\}) \leq g_M(\{p,q\})$ and $\lim_{M\rightarrow\infty}g_M(\{p,q\}) =0$.
	
	Let $n\in\N$ be arbitrary. Now we can compute the $n+2$-nd bit of the dyadic representation of $g_M(\{p,q\})$. Due to the channels, we obviously have $g_M(\{p,q\})\leq1$. Let $M_0=M_0(n)$ the smallest natural number such that the first $n+2$ bits of the dyadic representation of $g_{M_0}(\{p,q\})$ are zero. Then it holds for all $M\geq M_0$ that the $n+2$-nd bit of the dyadic representation of $g_M(\{p,q\})$ is zero as well due to the monotonic convergence. But this implies that $g_{M_0}(\{p,q\}) = \sum_{k=n+3}^\infty a_k\frac{1}{2^k}$, $a_k\in\{0,1\}$ and therewith $g_{M_0}(\{p,q\}) \leq \sum_{k=n+3}^\infty \frac{1}{2^k} = \frac{1}{2^{n+3}}\sum_{k=0}^\infty \frac{1}{2^k} = \frac{1}{2^{n+2}}$. With this we obtain
	\begin{equation*}
		0 \leq \Cu(M_0;\{p,q\}) - \Cl(M_0;\{p,q\}) < \frac{1}{2^{n+2}}.
	\end{equation*}
	For $M\geq M_0$ we have $\Cu(M;\{p,q\}) \leq \Cu(M_0;\{p,q\})$ and $\Cl(M;\{p,q\}) \geq \Cl(M_0;\{p,q\})$. Thus, for $M\geq M_0$ we obtain
	\begin{align*}
		0&\leq \Cu(M;\{p,q\}) - \Cl(M;\{p,q\}) 	\\
			&\leq \Cu(M_0;\{p,q\}) - \Cl(M;\{p,q\})	\\
			&\leq \Cu(M_0;\{p,q\}) - \Cl(M_0;\{p,q\}) \\
			&<\frac{1}{2^{n+2}}.
	\end{align*}
	Due to $C(\{p,q,s_0\})=\Cl(\{p,q\})=\lim_{M\rightarrow\infty}\Cl(M;\{p,q\})$ for all $s_0\in\sS$, we further have
	\begin{equation*}
		0 \leq \Cu(M;\{p,q\}) - C(\{p,q,s_0\}) < \frac{1}{2^{n+2}}.
	\end{equation*}
	The function $M_0=M_0(n)$ is effectively computable, since it is sufficient to run our algorithm until $a_{n+2}(g_{M_0}(\{p,q\}))=0$ is satisfied which completes the proof.
\end{proof}

\begin{remark}
	For every computable FSC $\{p,q,s_0\}$, $s_0\in\sS$, that satisfies $\Cu(\{p,q\})=\Cl(\{p,q\})$, we have $C(\{p,q,s_0\})\in\R_c$ for all $s_0\in\sS$, i.e., the capacity is a computable real number. This means that there exists an algorithm for the inputs $p,q$ that computes the desired approximation of the number $C(\{p,q,s_0\})$. In general, this algorithm does not depend recursively on the input $\{p,q\}$. Theorem \ref{the:comp} actually shows that this dependency must be non-recursive in general, since $C(\{p,q,s_0\})$ is not even Banach-Mazur computable in $\{p,q,s_0\}$.
\end{remark}

\begin{remark}
	If there exist $\{\hat{p},\hat{q}\}\in\sP_c\times\sQ_c$ and $s_0\in\sS$ such that $C(\{\hat{p},\hat{q},s_0\})\notin\R_c$, then this is the strongest form of non-computability, since then the value $C(\{\hat{p},\hat{q},s_0\})$ is not algorithmically computable although the parameters $\{\hat{p},\hat{q}\}\in\sP_c\times\sQ_c$ are computable real numbers. In \cite{BocheSchaeferPoor-2020-ICASSP-CompoundAlgorithmicPerspective,BocheSchaeferPoor-2020-TSP-CommunicationChannelUncertainty} it has been shown that there exist computable compound and averaged channels, where the state set is countably infinite, for which this phenomenon appears, i.e., there are computable compound and averaged channels such that its capacity is not a computable real number. This implies that for certain fixed computable compound or averaged channels, there exists no algorithm for the computation of the capacity. In \cite{BocheSchaeferPoor-2020-TSP-CommunicationChannelUncertainty} it has been further shown that such channels can be constructed based on binary symmetric channels. In addition to that, it has been shown that the achievability part, i.e., the coding part, cannot be constructive, i.e., there is no algorithm that can construct the corresponding encoder and decoder. This is particularly interesting to observe given the recent progress in polar codes that can construct algorithmically capacity-achieving encoder and decoder for fixed computable binary symmetric channels. The result in \cite{BocheSchaeferPoor-2020-TSP-CommunicationChannelUncertainty} on the other hand shows that this is no longer possible in general for compound and averaged channels. 

Some further comments are in order:
\begin{itemize}
	\item There are several definitions of computable functions which are not equivalent in general.
	\item The notion of Banach-Mazur computability is the weakest notion of computability.
	\item From a practical point of view, it is not clear if it makes sense to further weaken the requirements of Banach-Mazur computable functions.
	\item It is common sense that a computable function should map computable numbers from its domain to computable numbers within its value range. 
\end{itemize}
	
	To this end, it is interesting to observe that computable compound and averaged channels need not necessarily satisfy this basic requirement, cf. \cite{BocheSchaeferPoor-2020-TSP-CommunicationChannelUncertainty}, where computable channels are constructed whose capacity is a non-computable real number.
\end{remark}

\section{FSC Capacity as an Optimization Problem}
\label{sec:singleletter}

Let us go back one more time to Theorem \ref{the:comp} and its proof, where we analyzed the capacity function $C(\{p,q,s_0\}):\sP_c\times\sQ_c\times\sS\rightarrow\R$. We have shown that the capacity function $C(\{p,q,s_0\})$ is discontinuous for certain $s_0\in\sS$ and computable $p\in\sP_c$ and $q\in\sQ_c$.

\begin{theorem}
	\label{cor:discont}
	For all $|\sX|\geq2$, $|\sY|\geq2$, and $|\sS|\geq2$, the capacity function $C:\sP\times\sQ\times\sS\rightarrow\R$ is discontinuous.
\end{theorem}
\begin{proof}
	We consider the channels $p(y_n|x_n,0)$, $p(y_n|x_n,1)$, $\hat{q}(y_n|x_n,0)$, and $\hat{q}(y_n|x_n,1)$ as in \eqref{eq:p} and \eqref{eq:hatq}.
	
	Next, we consider $\{p,q_k,s_0\}$ for $k\geq1$ with 
	\begin{equation}
		q_k(s_n|x_n,0) \!=\! \begin{pmatrix}
			1\!-\!\frac{1}{k+1}  \\
			\frac{1}{k+1}
		\end{pmatrix},
		q_k(s_n|x_n,1) \!=\! \begin{pmatrix}
			\frac{1}{k+1}  \\
			1 \!-\! \frac{1}{k+1}
		\end{pmatrix}\!.
		\label{eq:qn}
	\end{equation}
	We observe that the FSC $\{p,q_k,s_0\}$, $s_0\in\sS$, $k\geq1$, as defined above is indecomposable. Further, $q_k$ is obviously computable, i.e., $q_k\in\sQc$, and further independent of $x_n\in\sX$.
	
	Next, we need a concept of distance. For $\pone,\ptwo\in\sPc$ and $\qone,\qtwo\in\sQc$ we define the distance between the FSCs $\{\pone,\qone,s_0\}$ and $\{\ptwo,\qtwo,s_0\}$ as
	\begin{equation}
		\begin{split}
		d(\{\pone,\qone,s_0\},\{\ptwo,\qtwo,s_0\}) \quad\qquad\qquad\qquad\\
			= \max_{x\in\sX}\sum_{y\in\sY}\big|\pone(y|x,s_0)-\ptwo(y|x,s_0)\big| \qquad \\
			+ \max_{x\in\sX}\sum_{s\in\sS}\big|\qone(s|x,s_0)-\qtwo(s|x,s_0)\big|.
		\end{split}
		\label{eq:dist}
	\end{equation}
	For FSCs as defined in \eqref{eq:p}-\eqref{eq:qn}, we have for any $s_0\in\sS$, $d(\{p,\hat{q},s_0\},\{p,q_k,s_0\}) = \frac{2}{k+1}$.
	
	Next, let us assume that $C(\{p,q,s_0\})$, $s_0\in\{0,1\}$, is Banach-Mazur computable on $\sP_c\times\sQc$. Then this would require that both capacities for $s_0=0$ and $s_0=1$ are continuous functions on $\sPc\times\sQc$. In particular, we must have $\lim_{k\rightarrow\infty}C(\{p,q_k,0\})=C(\{p,q,0\})$ and $\lim_{k\rightarrow\infty}C(\{p,q_k,1\})=C(\{p,q,1\})$.
	
	Since for all $k\in\N$ the FSC $\{p,q_k,s_0\}$, $s_0\in\sS$, is indecomposable, we have $C(\{p,q_k,0\})=C(\{p,q_k,1\})$ and further obtain
	\begin{align*}
		1 &= C(\{p,q,0\}) = \lim_{k\rightarrow\infty}C(\{p,q_k,0\}) = \lim_{k\rightarrow\infty}C(\{p,q_k,1\}) \\
			&= C(\{p,q,1\})=C_{\text{BSC}}(\epsilon)=1-H_2(\epsilon)<1
	\end{align*}
	which is a contradiction. Accordingly, at least one of the functions $C(\{p,q,0\})$ or $C(\{p,q,1\})$ must be discontinuous proving the desired result.
\end{proof}

This allows to obtain the following result.

\begin{theorem}
	\label{the:continuous}
	Let $|\sX|\geq2$, $|\sY|\geq2$, and $|\sS|\geq2$ be arbitrary. Then there is no natural number $n_0\in\N$ such that the capacity $C(\{p,q,s_0\})$ can be expressed as
	\begin{equation}
		\label{eq:singleletter_cont}
		C(\{p,q,s_0\}) = \max_{u\in\sU}F(u,p,q,s_0)
	\end{equation}
	with $\sU\subset\R^{n_0}$ a compact set and $F: \sU\times\sP\times\sQ\times\sS\rightarrow\R$ a continuous function.
\end{theorem}
\begin{proof}[Sketch of Proof]
	The result can be shown similarly as in \cite{BocheSchaeferPoor-2019-ISIT-IdentificationCorrelationAssisted}. The crucial observation is the following: To be able to express the capacity $C(\{p,q,s_0\})$ as in \eqref{eq:singleletter_cont}, the capacity necessarily needs to be a continuous function which cannot be the case by Corollary \ref{cor:discont}.
\end{proof}

\begin{remark}
	Theorem \ref{the:continuous} further immediately implies that the capacity $C$ cannot be expressed by a finite multi-letter formula. As a consequence, if $C$ can be described by entropic quantities, then this must be done via a corresponding sequence. Accordingly, the characterization via a limit of multi-letter expressions cannot be simplified and there is no closed form solution possible in general for the capacity of FSCs.
\end{remark}

\section{Discussion and Open Problems}
\label{sec:discussion}

In this paper, we have studied the capacity of FSCs and we have shown that the capacity function $C(\{p,q,s_0\})$ is not Banach-Mazur computable. As a consequence, the capacity does not depend recursively on the system parameters $\{p,q,s_0\}$ and it is impossible to algorithmically compute the capacity $C(\{p,q,s_0\})$. We have further shown that we cannot find tight lower and upper bounds on the capacity which are simultaneously computable continuous functions or Borel computable functions, respectively. This means that either the achievability or the converse (or both) must result in non-computable lower or upper bounds. It is not known which of them are actually non-computable for the FSC and, accordingly, the implications on the information theoretic approaches of the achievability and converse are unknown. Furthermore, the capacity is also shown to be non-approximable, i.e., it is impossible to approximate the capacity by computable sequences of computable functions for certain approximation errors.

For certain applications however, algorithmically computing the capacity of an FSC might be more than is actually needed. For example for resource allocation, it is often sufficient to know whether or not the current channel supports a certain quality-of-service (QoS) requirement $\lambda$. Accordingly, the following question is of interest: 
\vspace*{0.25\baselineskip}

\begin{tcolorbox}[colback=white,boxrule=0.125ex]
	{\bf Question 5:} Is there an algorithm (or Turing machine) that takes the FSC $\{p,q,s_0\}$ and the QoS requirement $\lambda>0$ as inputs and outputs ``\emph{yes}'' if $C(\{p,q,s_0\})>\lambda$ and outputs ``\emph{no}'' if $C(\{p,q,s_0\})<\lambda$?
\end{tcolorbox}

This is a \emph{decision} problem where the Turing machine decides whether or not an FSC supports a certain QoS requirement. Note that this Turing machine necessarily needs to stop for all possible inputs. However, it is not clear if problem is decidable and that such a Turing machine actually exists. In such a case, one may be inclined to weaken the question as follows:
\vspace*{0.25\baselineskip}

\begin{tcolorbox}[colback=white,boxrule=0.125ex]
	{\bf Question 6:} Is there an algorithm (or Turing machine) that takes the FSC $\{p,q,s_0\}$ and the QoS requirement $\lambda>0$ as inputs and stops if $C(\{p,q,s_0\})>\lambda$?
\end{tcolorbox}

This modified question asks whether or not it is \emph{semidecidable}. Here, the Turing machine must only stop and output the correct answer if the FSC supports the QoS requirement, i.e., $C(\{p,q,s_0\})>\lambda$. In the other case, it does not stop and runs forever. It is clear that one can pose this question also in the opposite way by requiring the Turing machine to stop only if $C(\{p,q,s_0\})<\lambda$.

There are several communication scenarios and channels whose capacity functions are not Turing computable, but their corresponding decision problems are semidecidable, cf. \cite{BocheSchaeferPoor-2019-WIFS-ResourceAllocation}. It is of interest to study such questions also for FSCs.

\bigskip

We want to conclude by coming back one more time to Kailath's work in information theory and the characterization of time-variant channels. In this case, the corresponding characterizations of capacities according to our results remain unknown. But as already mentioned in the introduction, there are further connections to the effective analysis and constructive mathematics. Here, the aim is to solve certain mathematical questions effectively, i.e., with the help of algorithms. 

Recently, impressive progress has been made in the theory of time-variant channels. For a detailed discussion we refer to \cite{WalnutPfanderKailath-2015-Birkhauser-CornerstonesSamplingOperatorTheory}. For example, progress in the design of test signals for channel identification \cite{LawrencePfanderWalnut-2005-JFAA-LinearIndependenceGaborSystems,PfanderWalnut-2006-TIT-MeasurementTimeVariantLinearChannels}, extension to the multiple-input multiple-output (MIMO) case \cite{PfanderWalnut-2016-TIT-SamplingReconstructionOperators,LeePfanderPohl-2019-TSP-SamplingReconstructionMIMOChannels}, stochastic channels \cite{PfanderZheltov-2014-ACHA-IdentificationStochasticOperators,PfanderZheltov-2014-TIT-SamplingStochasticOperators}, channels with unknown carrier \cite{HeckelBolcskei-2013-TIT-IdentificationSparseLinearOperators,PfanderWalnut-2016-TIT-SamplingReconstructionOperators}, constraints on the channel estimation \cite{LeePfanderPohlZhou-2019-LAA-IdentificationChannels}, and others. These results address many of the problems discussed in \cite{Kailath-1959-TechRep-SamplingLinearTimeVariantFilter} and provide solutions based on the classical analysis. In these works, methods such as distribution theory have been used that are not effective in general. This means that only the existence of certain strategies has been shown without the provision of effective algorithms or proofs. Note that this does not immediately exclude the possibility of a constructive characterization. But we want to note that in \cite{BocheMonich-2020-ICASSP-EffectiveApproximationBandlimitedSignals} computable absolutely integrable band-limited signals have been constructed, which are then also computable signals in $L^2(\R)$, for which the bandwidth $B(f)$ is not a computable real number. It is not clear if this yields the impossibility of effective characterizations of the results in the above mentioned works.

\section*{Acknowledgment}

Holger Boche would like to thank Volker Pohl for insightful discussions on time-continuous channels. He would like to further thank Robert Schober for interesting and fruitful discussions on the application of FSCs and time-continuous channels in molecular communication.

This work of H. Boche was supported in part by the German Federal Ministry of Education and Research (BMBF) within the national initiative for ``\emph{Molecular Communication (MAMOKO)}'' under Grant 16KIS0914 and in part by the German Research Foundation (DFG) within the Gottfried Wilhelm Leibniz Prize under Grant BO 1734/20-1 and within Germany's Excellence Strategy -- EXC-2111 -- 390814868. This work of R. F. Schaefer was supported in part by the BMBF within the national initiative for ``\emph{Post Shannon Communication (NewCom)}'' under Grant 16KIS1004 and in part by the DFG under Grant SCHA 1944/6-1. This work of H. V. Poor was supported by the U.S. National Science Foundation under Grants CCF-0939370, CCF-1513915, and CCF-1908308.

This paper was presented in part at the IEEE Information Theory Workshop (ITW), Visby, Sweden, Aug. 2019 \cite{BocheSchaeferPoor-2019-ITW-ComputabilityFSC} and in part at the National Research Meeting on Molecular Communications at the Friedrich-Alexander-Universit\"at Erlangen-N\"urnberg, Germany, Dec. 2018.



\address{Institute of Theoretical Information Technology\\
 Technische Universit\"at M\"unchen\\
Munich, Germany\\
\email{boche@tum.de}}

\address{Information Theory and Applications Chair\\
Technische Universit\"at Berlin\\	
Berlin, Germany\\
\email{rafael.schaefer@tu-berlin.de}}

\address{Department of Electrical Engineering\\
Princeton University\\
Princeton, NJ 08544, USA\\
\email{poor@princeton.edu}}

\end{document}